\documentclass[a4paper,11pt]{article}
\pdfoutput=1

\usepackage[utf8]{inputenc}
\usepackage{fullpage,amsfonts,bm,verbatim}
\usepackage{amsmath,amssymb,amsthm,mathtools}
\usepackage{mleftright}\mleftright
\usepackage[pagebackref,final,breaklinks=true]{hyperref}
\usepackage{enumitem}
\usepackage{algpseudocode}
\usepackage{algorithm}

\usepackage{tikz}
\usepackage{float}

\newtheorem{theorem}{Theorem}

\newtheorem{lemma}[theorem]{Lemma}
\newtheorem{claim}[theorem]{Claim}

\newtheorem{definition}[theorem]{Definition}

\mathchardef\mhyphen="2D

\newcommand{\ket}[1]{|#1\rangle}
\newcommand{\bra}[1]{\langle#1|}

\newcommand{\ketbra}[2]{|#1\rangle\! \langle #2|}

\newcommand{\opt}{\mbox{\rm OPT}}
\newcommand{\polylog}[1]{\mbox{\rm polylog}\left(#1\right)}

\newcommand{\R}{\mathbb{R}}

\newcommand{\eps}{\varepsilon}
\newcommand{\nrm}[1]{\left\lVert#1\right\rVert}
\def\01{\{0,1\}}

\newcommand{\diag}{\mathrm{diag}}

\newcommand{\bigO}[1]{\mathcal{O}\left( #1 \right)}

\newcommand{\bOt}[1]{\widetilde{\mathcal O}\left(#1\right)}

\newcommand{\term}[1]{\left(#1\right)}

\usepackage{ifdraft}
\ifdraft{\newcommand{\authnote}[3]{{\color{#3} {\bf  #1:} #2}}}{\newcommand{\authnote}[3]{}}

\begin{document}
\title{Quantum algorithms for zero-sum games}
\author{Joran van Apeldoorn\thanks{Supported by the Netherlands Organization for Scientific Research, grant number 617.001.351 and partially supported by QuantERA project QuantAlgo 680-91-034.}
  \and
  Andr\'as Gily\'en\thanks{Supported by ERC Consolidator Grant 615307-QPROGRESS and partially supported by QuantERA project QuantAlgo 680-91-034.}
}
\date{\small QuSoft, CWI and University of Amsterdam, the Netherlands. \\ e-mail: {\tt \{apeldoor,gilyen\}@cwi.nl}\\\vskip3mm\today}
\maketitle

\begin{abstract}
	We derive sublinear-time quantum algorithms for computing the Nash equilibrium of two-player zero-sum games, based on efficient Gibbs sampling methods.
We are able to achieve speed-ups for both dense and sparse payoff matrices at the cost of a mildly increased dependence on the additive error compared to classical algorithms.
In particular we can find $\eps$-approximate Nash equilibrium strategies in complexity $\bOt{\sqrt{n+m}/\eps^3}$ and $\bOt{\sqrt{s}/\eps^{3.5}}$ respectively, where $n\times m$ is the size of the matrix describing the game and $s$ is its sparsity.
Our algorithms use the LP formulation of the problem and apply techniques developed in recent works on quantum SDP-solvers.
We also show how to reduce general LP-solving to zero-sum games, resulting in quantum LP-solvers that have complexities $\bOt{\sqrt{n+m}\gamma^3}$ and $\bOt{\sqrt{s}\gamma^{3.5}}$ for the dense and sparse access models respectively, where $\gamma$ is the relevant ``scale-invariant'' precision parameter.
\end{abstract}
\section{Introduction}
A \emph{matrix game} is a two-player zero-sum game in which both players have only a finite number of pure strategies.
We label the moves for the first player (called Alice) by $[n]$ and the moves for the second player (called  Bob) by $[m]$.
If Alice plays $i\in[n]$ and Bob plays $j\in[m]$, then Alice gets a \emph{payoff} $A_{ij} \in [-1,1]$ and Bob gets payoff $-A_{ij}$.
Their individual goal is to get the highest payoff possible.
The payoff can be written in matrix form $A\in [-1,1]^{n\times m}$, hence the name \emph{matrix game}.
A game is called \emph{symmetric} if $m=n$ and $A = -A^T$, in other words, the payoff matrix is \emph{skew symmetric}.

If one of the players would always play the same move, then (for most games) it would be easy for the other player to win.
Hence a strategy will be randomized in general.
Let $\Delta^n$ be the set of all non-negative vectors in $\R^n$ that sum to $1$, i.e., the set of all probability distributions over $n$ elements.
Given a randomized strategy $x\in\Delta^n$ for Alice and a randomized strategy $y\in\Delta^m$ for Bob, the expected payoff for Alice is $x^TAy$.

Naturally Bob's goal is to minimize Alice's expected payoff; the best he can do is to assume that Alice plays the best strategy $x$ on her side and optimize his $y$ for that:
\[
   \min_{y \in \Delta^m} \max_{x \in \Delta^n} x^TAy
\]

We can write this as a linear program (LP) by noting that a linear function over the simplex (in this case $x\mapsto x^TAy$) is maximized on a vertex of the simplex:
\[
   \min_{y \in \Delta^m} \max_{x \in \Delta^n} x^TAy = \min_{y \in \Delta^m} \max_{i\in[n]} e_i^TAy
\]
which can be written as an LP:
\vspace{-2mm}
\begin{align}\label{eq:primal}
  \min \ \ \ &\lambda\\
  s.t.
\ \ \ & Ay \leq  \lambda e\nonumber\\
             & y\in\Delta^m\nonumber\\
             & \lambda\in [-1,1]\nonumber
\end{align}
where $e$ is the all-one vector.
Notice that since Alice's strategy $x$ is in $\Delta^n$ she can indeed not get a better expected value than $\lambda$.
In fact, the \emph{dual} of this LP is the LP for Alice, and due to strong duality this gives the same value! Hence we will call the optimal $\lambda^*$ the \emph{value} of the game. The corresponding strategies are called a \emph{Nash equilibrium}.
Notice that for a symmetric game the value is always $0$.
For completeness, let us also state the dual problem:
\begin{align}\label{eq:dual}
  \max \ \ \ &\lambda'\\
  s.t.
\ \ \ & A^T x \geq  \lambda' e\nonumber\\
             & x\in\Delta^n\nonumber\\
             & \lambda'\in [-1,1]\nonumber
\end{align}

We will call a strategy $y$ for Bob \emph{$\eps$-optimal} when $Ay\leq (v^*+\eps)e$, and similar for Alice.
Grigoriadis and Khachiyan~\cite{grigoriadis1995SubLinRndApxAlgMatrixGames} showed that a pair of $\eps$-optimal strategies can be found using a classical computation consisting of $\bigO{\log(n+m)/\eps^2}$ iterations, and $\bigO{n+m}/\log(n+m)$ steps per iteration\footnote{The authors show that the steps in each iteration can be highly parallelized.}.
In fact, a more careful analysis shows that an iteration can be implemented in $s$ operations, where $s$ is the maximal row and column sparsity of $A$.
Notice that this leads to a \emph{sub-linear} amount of work!
They show this by first converting any game to a symmetric game and then showing that symmetric games can be solved by a randomized algorithm.
In Section~\ref{sec:classical}, we give a proof of their results that directly applies to non-symmetric games.
In the same paper the authors also prove that any \emph{deterministic} algorithm would require at least $mn/2$ queries to the input.

As in recent work on convex optimization using \emph{quantum computers} \cite{brandao2016QSDPSpeedup,apeldoorn2017QSDPSolvers,brandao2017QSDPSpeedupsLearning,apeldoorn2018ImprovedQSDPSolving} we show that using a quantum computer a quadratic improvement in terms of the dimensions $n$ and $m$ can be achieved, at the expense of a slightly heavier dependence on the precision. In particular, in Section~\ref{subsec:qdense} we show that on a quantum computer $\bOt{\sqrt{n+m}/\eps^3}$ queries to the entries of $A$ and the same number of other gates suffice to implement the algorithm by Grigoriadis and Khachiyan~\cite{grigoriadis1995SubLinRndApxAlgMatrixGames}.
In Section~\ref{subsec:qsparse} we show that this can be improved to $\bOt{\sqrt{s}/\eps^{3.5}}$ for sparse matrices. Note that unlike in the aforementioned works, the classical algorithm that we speed up (and hence our quantum algorithm) does not depend on additional scale parameters, thus the achieved speed-ups seem more applicable in practice.

In Section~\ref{sec:reduction} we also show how to reduce general LP-solving to zero-sum games, resulting in our new general purpose quantum LP-solvers. This reduction introduces an extra dependence on the size of the primal and dual solutions in the complexity. However, the dependence on these parameters and on the approximation error is only cubic, where as an LP-solver obtained from the recent results on SDP-solving~\cite{apeldoorn2018ImprovedQSDPSolving} would have a fifth power dependence. Furthermore, we give the first quantum LP-solver which depends on the sparsity of the LP instead of on $n$ and $m$\footnote{The sparsity of an LP should not be confused with the sparsity parameter relevant in SDP solving. An LP that is written as an SDP will have SDP sparsity $1$ since all the matrices involved are diagonal. Our goal here is to get a dependence on the LP sparsity instead of a dependence on $n$ and $m$.}.

\paragraph{Notation.}
We write $\bm{e}$ for Euler's constant. For a vector $x\in\R^n$ we write $\bm{e}^x$ for the vector with entries $\bm{e}^{x_i}$.
Throughout the paper $i$ is an index in $[n]$ and  $j$ is an index in $[m]$.
We write $e_j$ for the $j$-th standard basis vector and $e$ for the all-one vector when the dimension is clear from context. We use notation $\bOt{T}$ as a shorthand for $\bigO{T\cdot\polylog{T\frac{nm}{\eps\delta}}}$.

\paragraph{Computational model.} When talking about the time complexity or gate complexity of a quantum algorithm then we assume that all two-qubit gates have unit cost. Furthermore we assume access to a classical-write / quantum-read random access memory at unit cost.
We only require such a memory consisting of $\bOt{1/\eps^2}$ bits, so not allowing such a gate will worse the gate complexity by at most a factor of $\bOt{1/\eps^2}$.

\section{Classical algorithm} \label{sec:classical}
In this section we will present the classical zero-sum game algorithm developed by Grigoriadis and Khachiyan~\cite{grigoriadis1995SubLinRndApxAlgMatrixGames}, with two alterations:
\begin{enumerate}
\item We give the algorithm and its proof without first reducing the problem to symmetric games.
In this way we lay more emphasis on the fact that this is a \emph{primal-dual} approach and on the connection to \emph{fictitious play}.
Furthermore, we hope that this view on the algorithm will allow for an easier application to other problems.
\item The algorithm by Grigoriadis and Khachiyan assumes the desired additive error $\eps$ is known from the start of the algorithm.
We present a version of the algorithm for which the additive error in the intermediate solutions uniformly decreases during the run of the algorithm, but we also consider a version corresponding to a fixed accuracy goal more similar to their results.
\end{enumerate}

\begin{algorithm}[H]
\caption{Main algorithm}\label{alg:main}
\begin{algorithmic}
  \State $x^{(0)} \gets 0 \in \R^n$ and $y^{(0)} \gets 0 \in \R^m$
  \For { $t = 1,2,\dots$}
     \State $\eta^{(t)} = \frac{1}{2\sqrt{t}}$ (alternatively in the fixed accuracy-goal version choose $\eta^{(t)} = \eps/4$)
     \State $u^{(t)} \gets -A^Tx^{(t)}$ and $v^{(t)} \gets Ay^{(t)}$
     \State $P^{(t)} \gets \bm{e}^{ u^{(t)}} = \bm{e}^{ -A^T x^{(t)}}$ and $Q^{(t)} \gets \bm{e}^{ v^{(t)}} = \bm{e}^{Ay^{(t)}}$
     \State $p^{(t)} \gets P^{(t)} / \nrm{P^{(t)}}_1$ and $q^{(t)} \gets Q^{(t)} / \nrm{Q^{(t)}}_1$
     \State Sample $a\sim p^{(t)}$ and $b\sim q^{(t)}$
     \State $y^{(t+1)} = y^{(t)}+\eta^{(t)} e_a$  and $x^{(t+1)} = x^{(t)}+ \eta^{(t)} e_b$
  \EndFor
\end{algorithmic}
\end{algorithm}
We start by proving that this algorithm is correct, before giving a bound on the complexity.
\begin{lemma}\label{lem:correctness}
  Let $\delta\in (0,1/3)$.
With probability at least $\delta$ Algorithm~\ref{alg:main} produces a sequence of $x^{(t)}$ and $y^{(t)}$ such that for all $t$ the intermediate solutions $x^{(t)}/\nrm{x^{(t)}}_1$ and $y^{(t)}/\nrm{y^{(t)}}_1$ are $\eps'$-optimal solutions with
  \[
    \eps' = \frac{2}{\sqrt{t}} \cdot\left(3\ln(t) + \ln(nm) + \ln(1/\delta) + 2\right).
  \]
  Let $\eps\in (0,1)$. If we run the algorithm with $\eta^{(t)}:=\eps/4$ instead, then, with probability at least $1-\delta$, the solutions are $\eps$-optimal after $T=\Big\lceil\frac{16\ln\left(\frac{nm}{\delta}\right)}{\eps^2}\Big\rceil$ iterations.
\end{lemma}

\begin{proof}
 Note that $x^{(t)}/\nrm{x^{(t)}}_1$ and $y^{(t)}/\nrm{y^{(t)}}_1$ are indeed probability distributions.
Hence $(x^{(t)}/\nrm{x^{(t)}}_1,\min_j (A^Tx^{(t)}/\nrm{x^{(t)}}_1)_j)$ and  $(y^{(t)}/\nrm{y^{(t)}}_1,\max_i (Ay^{(t)}/\nrm{y^{(t)}}_1)_i)$ are feasible points for the primal~\eqref{eq:primal} and dual~\eqref{eq:dual}.
Due to strong duality, to show that these solutions are $\eps$-optimal, it suffices to show that the difference between their values is at most $\eps$, that is
  \[
    \forall i\in[n] , j\in[m]: \term{Ay^{(t)}/\nrm{y^{(t)}}_1}_i - \term{A^Tx^{(t)}/\nrm{x^{(t)}}_1}_j \leq \eps
  \]
  To do so we consider the potential function
  \[
    \Phi(t) := \term{\sum_{j\in[m]} P_j^{(t)}}\term{\sum_{i\in[n]} Q_i^{(t)}} = \term{\sum_{j\in[m],i\in[n]} P_j^{(t)}Q_i^{(t)}}
  \]
  and show that this is bounded from above.
  In the beginning $\Phi(0) = nm$, moreover
	\vskip-3mm
  \begin{align*}
    \Phi(t+1) &= \term{\sum_{j\in[m]} P_j^{(t+1)}}\term{\sum_{i\in[n]} Q_i^{(t+1)}}\\
              &= \term{ \sum_{j\in[m]} \bm{e}^{ (-A^Tx^{(t)})_j - \eta^{(t)}A_{bj} }} \term{ \sum_{i\in[n]} \bm{e}^{ (Ay^{(t)})_i + \eta^{(t)} A_{ia} }}\\
              &= \term{ \sum_{j\in[m]} P^{(t)}_j \bm{e}^{- \eta^{(t)}A_{bj} }} \term{ \sum_{i\in[n]} Q_j^{(t)} \bm{e}^{  \eta^{(t)} A_{ia} }}\\
              &= \term{ \sum_{j\in[m]} p^{(t)}_j \nrm{P^{(t)}}_1 \bm{e}^{ -\eta^{(t)}A_{bj} }} \term{ \sum_{i\in[n]} q^{(t)}_i \nrm{Q^{(t)}}_1 \bm{e}^{  \eta^{(t)} A_{ia} }}\\
              &= \Phi(t) \term{ \sum_{j\in[m]} p^{(t)}_j  \bm{e}^{ -\eta^{(t)}A_{bj} }} \term{ \sum_{i\in[n]} q^{(t)}_i \bm{e}^{ \eta^{(t)} A_{ia} }}.
  \end{align*}
Taking the expectation over the sampling of $a$ and $b$ and working out the sums we get
  \begin{align*}
    \mathbb{E} [ \Phi(t+1)] &= \Phi(t)  \sum_{a\in[m]} \sum_{b\in[n]}  \sum_{j\in[m]} \sum_{i\in[n]}   p^{(t)}_a q^{(t)}_b p^{(t)}_j q^{(t)}_i \bm{e}^{ \eta^{(t)} (A_{ia} - A_{bj}) }\\
  &\leq \Phi(t)  \sum_{a\in[m]} \sum_{b\in[n]}  \sum_{j\in[m]} \sum_{i\in[n]}   p^{(t)}_a q^{(t)}_b p^{(t)}_j q^{(t)}_i \term{1+\eta^{(t)} \term{A_{ia}-A_{bj}} +3\term{\eta^{(t)}}^2},
  \end{align*}
  where we used the fact that for all $x\in[-1,1]$ it holds that $\bm{e}^x\leq 1+x+3x^2/4$
  , which implies $\bm{e}^{\eta^{(t)} \term{A_{ia}-A_{bj}}} - 1 - \eta^{(t)} \term{A_{ia}-A_{bj}}\leq\frac{3}{4}\left(\eta^{(t)} \term{A_{ia}-A_{bj}}\right)^2\leq 3\term{\eta^{(t)}}^2$.
  Now also observe that
  \[
    \sum_{a\in[m]} \sum_{b\in[n]}  \sum_{j\in[m]} \sum_{i\in[n]}   p^{(t)}_a q^{(t)}_b p^{(t)}_j q^{(t)}_i = 1
  \]
  and that all the $A_{ia}-A_{bj}$ terms cancel against another $A_{bj}-A_{ia}$ term with the same $p^{(t)}_a q^{(t)}_b p^{(t)}_j q^{(t)}_i$ coefficient.
  Hence
  \[
    \mathbb{E} [ \Phi(t+1)] \leq \Phi(t) \term{1+3\term{\eta^{(t)}}^2 }
  \]
  and by taking the expectation on both sides and expanding the recursion we get
  \[
    \mathbb{E} [ \Phi(t)]  \leq \Phi(0) \prod_{\tau=0}^{t-1} \term{1+3\term{\eta^{(t)}}^2} \leq nm \bm{e}^{ 3\sum_{\tau=0}^{t-1} \term{\eta^{(\tau)}}^2}.
  \]
  By Markov's inequality, with probability at least $1-\delta^{(t)}$ we have that
  \[
    \Phi(t)  \leq \frac{nm}{\delta^{(t)}} \bm{e}^{ 3\sum_{\tau=0}^{t-1} \term{\eta^{(\tau)}}^2}.
  \]
  Since $\Phi(t)$ is the sum of positive terms, each term itself is smaller than the sum.
  It follows that for all $i\in [n],j\in [m]$
  \[
    P_j^{(t)}Q_i^{(t)} = \bm{e}^{ \term{Ay^{(t)}}_i - \term{A^Tx^{(t)}}_j } \leq \frac{nm}{\delta^{(t)}} \bm{e}^{ 3\sum_{\tau=0}^{t-1} \term{\eta^{(\tau)}}^2}.
  \]
  Taking the logarithm on both sides, and dividing by $\nrm{x^{(t)}}_1 = \nrm{y^{(t)}}_1 = \sum_{\tau=1}^t \eta^{(t)}$ we get that
  \begin{equation}\label{eq:bound}
    \forall i\in[n] , j\in[m]: \term{Ay^{(t)} / \nrm{y^{(t)}}_1 }_i - \term{A^Tx^{(t)} / \nrm{x^{(t)}}_1 }_j \leq \frac{\ln\left(\frac{nm}{\delta^{(t)}}\right)+3\sum_{\tau=1}^t \term{\eta^{(t)}}^2 }{\sum_{\tau=1}^t \eta^{(t)}}.
  \end{equation}
  
  Until now every step works for both choices of $\eta^{(t)}$. First we finish the proof of the infinitely running version of the algorithm, where we choose $\delta^{(t)}:=\frac{\delta}{2t^2}$.
  Using the bounds $\sum_{\tau=1}^t \term{\eta^{(t)}}^2  = \frac{1}{4}\sum_{\tau=1}^t \frac{1}{\tau} \leq \frac{\ln(t)+1}{4}$ 
   and $\sum_{\tau=1}^t \eta^{(t)}  = \frac{1}{2}\sum_{\tau=1}^t \sqrt{\frac{1}{\tau}} \geq \sqrt{t}/2$ 
we find that, with probability at least $1-\frac{\delta}{2t^2}$, we have for \eqref{eq:bound}:
\begin{align*}
     \forall i\in[n] , j\in[m]:\term{Ay^{(t)} / \nrm{y^{(t)}}_1}_i-\term{A^Tx^{(t)} / \nrm{x^{(t)}}_1 }_j 
     &\leq \frac{\ln( 2\delta^{-1} t^2 nm )+(\ln(t)+1)}{\sqrt{t}/2} \\
     &\leq \frac{2}{\sqrt{t}} \cdot\left(3\ln(t) + \ln(nm) + \ln(1/\delta) + 2\right).
\end{align*}

  Taking the union bound over the error probabilities we have that the total error probability over all iterations is less than $\sum_{t\in \mathbb{N}} \frac{\delta}{2t^2} = \frac{\delta}{2} \cdot \frac{\pi^2}{6} \leq \delta$.
  
  Now we finish the analysis of the fixed-error version of the algorithm choosing $\eta^{(t)}:=\eps/4$ and $\delta^{(t)}:=\delta$. We can then, with probability at least $1-\delta$, bound \eqref{eq:bound} as
	\begin{align*}
	\forall i\in[n] , j\in[m]:\term{Ay^{(t)} / \nrm{y^{(t)}}_1 }_i - \term{A^Tx^{(t)} / \nrm{x^{(t)}}_1 }_j 
	&\leq \frac{\ln\left(\frac{nm}{\delta}\right)+3t\eps^2/16 }{t\eps/4}=\frac{3}{4}\eps + \frac{4\ln\left(\frac{nm}{\delta}\right)}{t\eps}.
	\end{align*}   
	For $t\geq \frac{16\ln\left(\frac{nm}{\delta}\right)}{\eps^2}$ this can be further upper bounded by $\eps$.
\end{proof}

Once we get the solutions we can estimate the value of the game by simply playing the games with the corresponding randomized strategies. Since in each run the game has bounded value, by Chernoff's bound we get the following:
\begin{claim}
	Given a pair of strategies $x,y$, let us take $k=\bigO{1/\eps^2}$ independent samples $i_1,\ldots,i_k$ from $x$ and similarly $j_1,\ldots,j_k$ from $y$.
	Then with high probability $\sum_{\ell=1}^k A_{i_\ell,j_\ell}/k$ is an $\eps$-approximate estimate of the value of the game corresponding to these strategies.
\end{claim}
\noindent This clearly shows that obtaining the approximate solutions via Algorithm~\ref{alg:main} dominates the complexity of approximately computing the corresponding value.

When access to $A$ is given by an oracle which can be queried for any matrix element of $A$, then the following lemma gives a bound on the cost of the algorithm.
\begin{lemma}\label{lemma:DenseClassical}
  Algorithm~\ref{alg:main} can be implemented to find an $\eps$-optimal pair of solutions on a classical randomized computer with probability at least $1-\delta$ in $\frac{16}{\eps^2} \ln\left(\frac{nm}{\delta}\right)$ iterations, using
  $\bigO{n+m}$ time per iteration and $n+m$ queries to the entries of $A$.
\end{lemma}
\begin{proof}
 The iteration bound follows from Lemma~\ref{lem:correctness}.
 In the rest of the proof we drop the $(t)$ subscript for ease of notation. Since in each iteration only one entry of $x$ changes, the change in $u$ can be computed using $m$ queries. It then requires $\bigO{m}$ work to compute $P$ and $p$. In the same amount of time the cumulative distribution corresponding to $p$ can be calculated. By generating a random number between $0$ and $1$ up to high precision and performing binary search on the cumulative distribution we can sample $a$ from $p$. A similar approach works for $y$, $v$, $Q$, $q$ and $b$.
\end{proof}
When $A$ is given by a \emph{sparse oracle} which also allows querying for any $j$ the location of the $j$-th non-zero entry in each row and column, then a further speedup is possible.

\begin{lemma}\label{lemma:SparseClassical}
Let $s$ be the maximum number of non-zero entries in a row, and $d$ the maximum number of non-zero entries in a column of $A$.
 Algorithm~\ref{alg:main} can be used to find an $\eps$-optimal pair of solutions on a classical randomized computer with probability at least $1-\delta$ in $\frac{16}{\eps^2} \ln\left(\frac{nm}{\delta}\right)$ iterations,
  $\bigO{s\ln(m) +d\ln(n)}$ time per iteration and $s+d$ queries to a sparse oracle for $A$.
\end{lemma}
\begin{proof}
  The iteration bound follows from Lemma~\ref{lem:correctness}. 
  In the rest of the proof we drop the $(t)$ subscript for ease of notation.
  Now, since there is a sparse oracle for the rows of $A$, the change in $u$ can be computed using $s$ queries and $\bigO{s\ln(m)}$ time.
  Hence, also the multiplicative change in $P$ can be calculated in $\bigO{s\ln(m)}$ time.

  Now, instead of keeping $P$ as a list, we keep it as a tree with the sum of the branches stored at each leaf, see Figure~\ref{fig:tree} for an example.
Now, updating one leaf can be done by walking the tree upward and changing each node as required, which takes $\bigO{\ln(m)}$ time.
Hence, $P$ can be updated in $s\ln(m)$ time.
Given a tree structure for the vector $P$, sampling from $P/\nrm{P}$ can easily be implemented: walk down the tree, going left or right with probabilities proportional to the node values.
Finally we note that there is no need to initialize a full-sized empty tree, the tree can be built during the run of the algorithm.

  A similar approach works for $y$, $v$, $Q$, $q$ and $b$ and the lemma statement follows.
\end{proof}

\begin{figure*}[tbp]
  \begin{center}
\begin{tikzpicture}[level 1/.style={sibling distance=10em},level 2/.style={sibling distance=5em},
  every node/.style = {shape=circle, rounded corners,
    draw, align=center,
    top color=white, bottom color=blue!20}]]
  \node {10}
  child { node {3}
      child { node {1} }
      child { node {2} }
  }
    child { node {7}
      child { node {3} }
      child { node {4} }
    };
  \end{tikzpicture}
  \end{center}
\caption{Tree structure for the vector $(1,2,3,4)$}
\label{fig:tree}
\end{figure*}
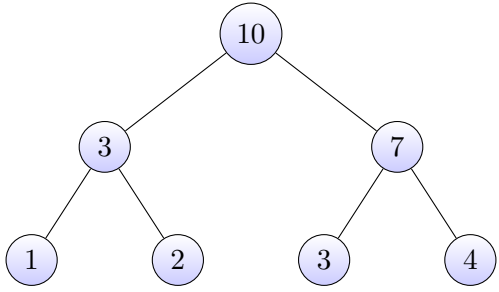

\section{Quantum implementation of the algorithm}
\label{sec:quantum}

In the quantum case we aim at a sublinear-time algorithm in $m$ and $n$, which strictly speaking cannot even read through a single column or row of $A$.
For this we need to assume that we can make quantum queries\footnote{A quantum oracle is a unitary that acts analogously to the classical oracle on basis states, but allows making queries in superposition; we also assume access to the inverse of such a unitary. For more details see, e.g.,~\cite{berry2012BlackHamSimUnitImp}. } to the oracles describing $A$. 
Also observe that the algorithm always maintains $x^{(t)}$ and $y^{(t)}$ that have at most $t$ non-zero elements.
Therefore we can efficiently store $x^{(t)}$ and $y^{(t)}$ during the algorithm and we will only work implicitly with the non-sparse vectors $v,u,p$ and $q$.
We store $x^{(t)}$ in a tree-like data structure in QCRAM, similar to the structure discussed at the end of the last section, see in Figure~\ref{fig:tree}.
This enables us to query elements of $x$ and sample from $i\in [n]$ with probability $\frac{x_i}{\nrm{x}_1}$, both in time $\bigO{\log(n)}$.
Moreover, it enables the preparation of the quantum state $\sum_{i\in[n]}\sqrt{x_i/\nrm{x}_1}\ket{i}$ in time $\bigO{\log(n)}$, as discussed in~\cite{grover2002SuperposEffIntegrProbDistr,kerenidis2016QRecSys}.
Implementing one iteration of Algorithm~\ref{alg:main} essentially boils down to efficient Gibbs-sampling, i.e., implementing sampling from the distribution $\bm{e}^{x^T A}/{\lVert \bm{e}^{x^T A}\rVert}_1$.

\begin{definition}[Classical Gibbs distribution]\label{def:Gibbs}
	For a vector $v\in \R^n$ let $\bm{e}^v$ denote the vector whose $i$-th coordinate is $\bm{e}^{v_i}$. Then $G(v):=\frac{\bm{e}^v}{\nrm{\bm{e}^v}_1}$ denotes the \emph{Gibbs distribution} corresponding to $v$.
\end{definition}

First we broadly describe the classical analogue of the quantum process used in Sections~\ref{subsec:qdense}.
Following this we will state a few useful technical results in Section~\ref{subsec:technical} and then in Sections~\ref{subsec:qdense}-\ref{subsec:qsparse} we finally show how to implement the process quadratically more efficiently on a quantum computer.

We will focus on the $t$-th iteration of Algorithm~\ref{alg:main}.
We assume that $x^{(t)}$ is $t$-sparse and is stored in the tree data structure, our goal is to sample from $G(x^TA)$.
Since $x$ is $t$-sparse we can compute a single element of $u = x^TA$ using $\bOt{t}$ steps.
Since we have a procedure to calculate $u_j$, we can find the maximal element $u_j$ using $\bigO{m}$ calls to this procedure.
Call this maximum $u_{\max}$.
We will sample from $G(x^T A - u_{\max}e) = G(x^TA)$ since this allows us to use that $\bm{e}^{u_j - u_{\max}}\leq 1$ for all $j$.
To do this sampling we will use rejection sampling: sample $j\in [m]$ uniformly at random, then with probability $\bm{e}^{u_j - u_{\max}}$ output $j$ and with probability $1-\bm{e}^{u_j - u_{\max}}$ output $\bot$.
If we would post-select on the outcome not being $\bot$, then we would have sampled from $G(x^T A - u_{\max}e) = G(x^TA)$.
Hence we repeat this procedure an expected $\bigO{\frac{m \bm{e}^{u_{\max}}}{\sum_{j=1}^m \bm{e}^{u_j}}} \leq \bigO{m}$ times until we get an output other than $\bot$, resulting in complexity $\bOt{mt}$.

All the steps described above have quadratically faster quantum counterparts.
Estimating a single value $u_j$ can be done via amplitude estimation, finding the maximum can be achieved using the maximum-finding algorithm by Dürr and Hoyer~\cite{durr1996QMinimumFinding} and rejection sampling can be done in $\bigO{\sqrt{m}}$ steps via amplitude amplification.


\subsection{Technical Lemmas}\label{subsec:technical}
In order to describe our quantum algorithms succinctly we introduce some technical results about block-encodings.
A unitary $U$ is an $a$-qubit block-encoding of $A$ if the top-left block of the unitary~$U$ is $A$:
\[
  A= \left(\bra{0}^{\otimes a}\otimes I\right) U \left(\ket{0}^{\otimes a}\otimes I\right)\Longleftrightarrow U= \left[\begin{array}{cc} A & .
\\ .
& .
   \end{array}\right].
\]
One can think of this as a probabilistic implementation of the linear transformation $A$: given an input state $\ket{\psi}$, applying the unitary $U$ to the state $\ket{0}^{\otimes a}\ket{\psi}$, measuring the first $a$-qubit register and post-selecting on the $\ket{0}^{\otimes a}$ outcome, we get a state $\propto A\ket{\psi}$ in the second register.
We will often just say that we apply the matrix $A$ to a quantum state $\ket{\psi}$ by which we mean that we take the state $\ket{0}^{\otimes a}\ket{\psi}$, where $\ket{0}^{\otimes a}$ are fresh ancilla qubits, and we apply the block-encoding $U$ to the state $\ket{0}^{\otimes a}\ket{\psi}$, and treat the ancilla states $\ket{0}^{\otimes a}$ as an indicator of success.
For simplicity we will denote $\ket{0}^{\otimes a}$ by $\ket{\bar{0}}$, and in all our applications we will have $a=\bigO{\log(m+n)}$.
We will use the following result about transforming eigenvalues of block-encodings~\cite{low2017HamSimUnifAmp,gilyen2018QSingValTransf}:
\begin{theorem}[Polynomial eigenvalue transformation~{\cite[Theorem~56]{gilyen2018QSingValTransf}}]\label{thm:polyTransf}
	Suppose that $U$ is an $a$-qubit block-encoding of a Hermitian matrix $A$, and $P\in\R[x]$ is a degree-$d$ polynomial satisfying that
	\begin{enumerate}[label=(\roman*)]
		\item for all $x\in[-1,1]\colon$ $|P(x)|\leq \frac12$, or \label{it:noParity}
		\item for all $x\in[-1,1]\colon$ $|P(x)|\leq 1$ and $|P(x)|=|P(-x)|$.
	\end{enumerate}
	Then there is a quantum circuit $\tilde{U}$, which is an $(a+2)$-qubit block-encoding of $P(A)$, and which consists of $d$ applications of $U$ and $U^\dagger$, (and in case~\ref{it:noParity} a single application of controlled-$U^{\pm1}$) and $\bigO{(a+1)d}$ other one- and two-qubit gates.\footnote{Given the polynomial $P$ we can compute the gates for implementing an approximate polynomial $\tilde{P}$ such that $|P-\tilde{P}|\leq \xi$ for all $x\in[-1,1]$ in time $\bigO{d^3\polylog{d/\xi}}$ as shown by Haah~\cite{haah2018ProdDecPerFuncQSignPRoc}.}
\end{theorem}

Finally we invoke some polynomial approximation results which are explicitly or implicitly proven in~\cite{gilyen2018QSingValTransf}.
\begin{lemma}[Polynomial approximations]\label{lem:polyApx}
	Let $\beta\geq 1$, $\xi\leq 1/2$ and $t\in[0,1-\frac{1}{\beta}]$. There exist polynomials $\tilde{P},\tilde{Q},\tilde{R}$ such that
	\begin{itemize}
		\item $\forall x\in[-1,1]\colon|\tilde{P}(x)|\leq \frac12$ and for all $x\in[-1,0]\colon$ $|\tilde{P}(x)-\bm{e}^{\beta x}/4|\leq \xi$,
		\item $\forall x\in[-1,1]\colon |\tilde{Q}(x)|=|\tilde{Q}(-x)|\leq 1$, $\forall x\in [0,t]\colon |\tilde{Q}(x)|\leq\xi$, $\forall x\in [t+\frac{1}{\beta},1]\colon \tilde{Q}(x)\geq1-\xi$,
		\item $\forall x\in[-1,1]\colon |\tilde{R}(x)|=|\tilde{R}(-x)|\leq 1$, and $\forall x\in [\frac{1}{2\beta},1]\colon |\tilde{R}(x)-\frac{1}{\sqrt{4\beta x}}|\leq\xi$,
	\end{itemize}
	moreover $\deg(\tilde{P})$, $\deg(\tilde{Q})$, $\deg(\tilde{R})= \bigO{\beta\log(1/\xi)}$.
\end{lemma}

\subsection{Dense matrices}\label{subsec:qdense}

\begin{lemma}[Sampling from the Gibbs distribution of a linear combination of rows]\label{lem:GibbsDense}
 	Suppose that $x\in \R^n$ is stored in a data structure as in Figure~\ref{fig:tree}, such that\footnote{If $\nrm{x}_1\leq 1$ then the Gibbs distribution is very close to uniform so we can Gibbs sample in time $\bOt{1}$.} $\nrm{x}_1\leq \beta$ for some $\beta\geq 1$. If we have quantum query access to the matrix elements of $A\in [-1,1]^{n\times m}$, then we can sample from a distribution that is $\delta$-close to $G(x^T A)$ in total variation distance, in query and time complexity $\bOt{\beta\sqrt{m}}$ on a quantum computer with QCRAM.
\end{lemma}
\begin{proof}
	Using the tree data structure we can implement a unitary in time $\bOt{1}$ that acts as
\begin{align*}
	\ket{0}\ket{\bar{0}}
	&\mapsto \ket{0}\sum_{i\in[n]}\sqrt{x_i/\beta}\ket{i} + \ket{1}\ket{\mathrm{garbage}},
\end{align*}
where $\ket{\mathrm{garbage}}$ is some subnormalized quantum state on a (possibly multi-qubit) register.
Similarly, with a single additional query we can implement the unitary $V$ acting as\footnote{In case of a complex matrix entry $A_{ij}$ any consistent choice of $\sqrt{A_{ij}}$ works, e.g., we can set $\sqrt{-1}$ to $\bm{i}$.}
\begin{align*}
	V \colon \ket{0}\ket{0}\ket{\bar{0}}\ket{j}
	&\mapsto \ket{0}\left(\ket{0}\sum_{i\in[n]}\sqrt{x_i A_{ij}/\beta}\ket{i} + \ket{1}\ket{\mathrm{garbage}'}\right)\ket{j},
\end{align*}
where the third and fourth registers have $\bOt{1}$-qubits.
Let $\mathrm{SWAP}_{12}$ be the gate that swaps the first two qubits.
Now observe that for $U:=V^\dagger \left(\mathrm{SWAP}_{12}\otimes I \right)V$ we have
\[
  (\bra{00\bar{0}}\otimes I)U(\ket{00\bar{0}}\otimes I)=\diag(x^T A/\beta),
\]
i.e., $U$ is a block-encoding of the Hermitian matrix $\diag(x^T A/\beta)$.
We can estimate the value of $u_j=\sum_{i\in[n]}x_i A_{ij}$ to additive error say $1/2$ with $\bOt{\beta}$ uses of $U$ and $U^{-1}$ by amplitude estimation~\cite{brassard2002AmpAndEst}.
Therefore, with $u_{\max}:=\max_{j\in[m]}u_{j}$, we can also find a constant additive approximation $\tilde{u}_{\max}\in[u_{\max},u_{\max}+1]$ in time $\bOt{\beta\sqrt{m}}$ with high probability using generalized quantum minimum/maximum-finding~\cite{apeldoorn2017QSDPSolvers}. We can boost the success probability to $1-\frac{\delta}{3}$ by $\bigO{\log(1/\delta)}$ repetitions.
Using simple block-encoding techniques~\cite{gilyen2018QSingValTransf} we can then also implement $U'$, a block-encoding of $M := \diag(x^T A - \tilde{u}_{\max} e)/(2\beta)$, with a constant overhead.

Now we are ready to implement the rejection sampling.
We first prepare the uniform superposition $\frac{1}{\sqrt{m}}\sum_{j\in[m]}\ket{j}$.
 Then ideally we would like to apply the map
\begin{align*}
  \ket{0}\ket{j} &\mapsto
  \sqrt{\bm{e}^{u_j-\tilde{u}_{\max}}}\ket{0}\ket{j}+\ket{1}\ket{\mathrm{garbage}''}\\
                 &= \ket{0}\bm{e}^{\beta \frac{ \diag( x^T A - \tilde{u}_{\max} e)  }{2\beta}}\ket{j}+\ket{1}\ket{\mathrm{garbage}''}.\\
  &= \ket{0} \bm{e}^{\beta M} \ket{j} + \ket{1}\ket{\mathrm{garbage}''}.
\end{align*}
to this uniform superposition.

We implement a good approximation of the above by replacing the function $\bm{e}^{\beta z}/4$ with an approximating polynomial $\tilde{P}(z)$, and then using eigenvalue transformation Theorem~\ref{thm:polyTransf}.
For this we use an approximation polynomial $\tilde{P}(z)$ as in Lemma~\ref{lem:polyApx}, with $\xi=\frac{\delta}{12em}$.
We can then implement the map
$$\ket{0}\ket{\bar{0}}\ket{j}\mapsto \ket{0}\ket{\bar{0}}\tilde{P}\left(\diag(x^T A - \tilde{u}_{\max}e)/ 2\beta\right)\ket{j}/4+\ket{1}\ket{\mathrm{garbage}'''},$$
using $\bOt{\deg(\tilde{P})}=\bOt{\beta\log(1/\delta)}$ steps as shown by Theorem~\ref{thm:polyTransf}.

Finally, we obtain a sample with probability at least $1-\frac{\delta}{9}$ using $\bigO{\sqrt{m}\log(1/\delta)}$ rounds of amplitude amplification.
Note that due to the $\delta/(12em)$ accuracy, we sample from a distribution $\bigO{\delta}$-close to $G(x^T A)$ in total variation distance.
\end{proof}

\begin{theorem}\label{thm:DenseQuantum}
  Algorithm~\ref{alg:main} can be used to find an $\eps$-optimal pair of solutions with probability at least $1-\delta$ in $\frac{16}{\eps^2} \ln\left(\frac{nm}{\delta}\right)$ iterations. On a quantum computer with QCRAM the $t$-th iteration can be implemented in $\bOt{(1+t\eps)\sqrt{n+m}}$ time and the same number of quantum queries to the entries of $A$, leading to a total time $\&$ query complexity of $\bOt{\sqrt{n+m} / \eps^3}$.
\end{theorem}
\begin{proof}
 This follows from Lemma~\ref{lem:correctness} and Lemma~\ref{lem:GibbsDense}, by setting the additive error to $\bigO{\eps^2\delta/\ln(nm/\delta)}$ in the latter.
\end{proof}

\subsection{Sparse matrices}\label{subsec:qsparse}
Now we describe our result for the case when $A$ has sparse rows and columns.
We state the result with $s$ denoting an upper bound on the sparsity of both the rows and the columns.
Note that if separate bounds are known for each, then $s$ is simply the maximum of the two bounds.

\begin{lemma}[Gibbs sampling from a linear combination of sparse rows]\label{lem:GibbsSparse}
	Suppose $x\in \R_{\geq 0}^n$ has at most $t$ non-zero entries and is stored in a data structure as in Figure~\ref{fig:tree}, and $\nrm{x}_1\leq \beta$ for $\beta\geq 1$. If we have quantum query access to a sparse oracle of $A\in [-1,1]^{n\times m}$ having $s$-sparse rows and columns, then we can sample from a distribution that is $\delta$-close to $G(x^T A)$ in total variation distance, in query and time complexity $\bOt{\beta^\frac{3}{2}\sqrt{s}}$ on a quantum computer with QCRAM.
\end{lemma}
\begin{proof}
	We can assume without loss of generality that $\beta s\leq m/4$, otherwise the statement follows from Lemma~\ref{lem:GibbsDense}.
Let us define $w\in\R^m$ by $w_j:=\sum_{i\in[n]}x_i |A_{ij}|\geq u_j$.
The main idea is the following: we can implement Gibbs sampling efficiently for $j$'s where $w_j\leq 1$, this is because $w_j\leq 1$ for at least half of the $j$'s (as $\sum_{j\in[m]}w_j\leq \beta s \leq m/2$), this implies that for these $j$ we have $|u_j|\leq 1$ and hence the distribution $G(u)$ is quite uniform on these positions.
On the other hand Gibbs sampling $j$'s for which $w_j\geq 1/2$ is also easy, this is because we can find such $j$'s efficiently by first sampling an $i$ with probability $x_i/\beta$ and then sampling a $j$ with probability $\frac{|A_{ij}|}{s}$.
Since the two regimes cover every $j \in [m]$ we can Gibbs sample a $j$ by combining the two sampling procedures, in a similar fashion to~\cite{brandao2017QSDPSpeedupsLearning,apeldoorn2018ImprovedQSDPSolving}.
Hence the proof is structured as follows: first we show how to approximate $u_{\max}$ to constant additive error, then we show how to distinguish the two regimes.
Following this we show how to handle the small $w_j$'s and then the large $w_j$'s.
We conclude by showing how to apply these procedures together.
	
	To approximate $u_{\max}$ recall that using the tree data structure, as in Lemma~\ref{lem:GibbsDense}, we can implement a block-encoding $U$ of $\diag(x^T A/\beta)$ with a single query to $A$ in time $\bOt{1}$.
    We can estimate the value of $u_j=\sum_{i\in[n]}x_i A_{ij}$ to additive error $1/2$ with $\bOt{\beta}$ uses of $U$.
	As in Lemma~\ref{lem:GibbsDense} we use the tree data structure to implement a unitary acting as
	\begin{align*}
	\ket{0}\ket{\bar{0}}
	&\mapsto \ket{0}\sum_{i\in[n]}\sqrt{x_i/\beta}\ket{i} + \ket{1}\ket{\mathrm{garbage}}.
	\end{align*}
    in time $\bOt{1}$.
	With a single additional query, we can implement a unitary acting as
	\begin{align*}
	\ket{0}\ket{\bar{0}}\ket{\bar{0}}
	&\mapsto \ket{0}\sum_{i\in[n]}\sum_{j\colon A_{ij}\neq 0}\sqrt{x_i |A_{ij}|/(\beta s)}\ket{i}\ket{j} + \ket{1}\ket{\mathrm{garbage}'},
	\end{align*}
    We will call the resulting state $\ket{\psi}$.
    Now, we only need to approximate the maximum $u_{\max}$ if it is larger than $1$, otherwise we just set $\tilde{u}_{\max} = 1$, hence we start by assuming that it is greater than~$1$.
	Observe that if $u_{j}> 1$, then $\left(\ketbra{0}{0}\otimes I\otimes\ketbra{j}{j}\right)\ket{\psi}$ has squared norm $w_j/(\beta s)\geq u_j/(\beta s)> 1/(\beta s)$.
This means that the probability of getting outcome $\ket{0}$ in the first register and $\ket{j}$ in the last register when measuring $\ket{\psi}$ is at least $1/(\beta s)$.
	We can find a constant additive approximation $\tilde{u}_{\max}\in[u_{\max},u_{\max}+1]$ of the maximum value by using generalized quantum minimum/maximum-finding~\cite{apeldoorn2017QSDPSolvers} with  $\ket{\psi}$ as the initial superposition over $j$'s and by estimating $u_j$ using $U$.
Thus in total it takes time $\bOt{\beta^{\frac{3}{2}} \sqrt{s}}$ to find a constant approximation of the maximum with high probability, and as in Lemma~\ref{lem:GibbsDense} we can boost the success probability to $1-\delta/3$ with $\bigO{\log(1/\delta)}$ repetitions.
If all $u_j$ are less than $1$, then the algorithm will return an estimate $\tilde{u}_{\max}$ that is less than $1$ and hence we can detect this.
	
	Now to distinguish the two cases for $w_j$, consider the unitary $W$ that acts as
	\begin{align*}
		W\colon \ket{0}\ket{0}\ket{\bar{0}}\ket{j}
		&\mapsto \ket{0}\left(\ket{0}\sum_{i\in[n]}\sqrt{x_i|A_{ij}|/\beta}\ket{i} + \ket{1}\ket{\mathrm{garbage''}}\right)\ket{j},
	\end{align*}
	which can be implemented using two queries and $\bOt{1}$ gates.
Observe that the unitary $V:=W^\dagger(\mathrm{SWAP}_{12}\otimes I\otimes I)W$ is a block-encoding of the Hermitian matrix $\diag(w/\beta)$.
Let $\xi$ be a precision parameter that will be set later. Using Lemma~\ref{lem:polyApx} and Theorem~\ref{thm:polyTransf} we can implement a unitary $V'$ that is a block-encoding of $\diag(\tilde{ Q}(w))$ such that 
	\begin{align*}
	\tilde{ Q}(w_j)^2\in \begin{cases}
	[0,\xi/2], & \text{if } w_j\leq \frac{1}{2} \\
	[1-\xi/2,1] & \text{if } w_j\geq 1,
	\end{cases}
	\end{align*}
	using\footnote{\label{foot:QSVT}A similar unitary can be implemented with complexity $\bigO{\sqrt{\beta}\log(1/\xi)}$ by quantum singular value transformation~\cite{gilyen2018QSingValTransf} of $W$, however this operation is not the bottleneck so we do not  optimize it heavily.} $\bOt{\beta\log(1/\xi)}$ time and queries.
Now there is a $Q\colon [0,\beta]\to [0,1]$ such that $\tilde{Q}^2$ is an additive $\frac{\xi}{2}$-approximation of the function $\tilde{Q}^2$, $Q$ is zero on $[0,1/2]$, and $Q$ is $1$ on $[1,\beta]$. $Q$ should be thought of as the ``idealized'' version of $\tilde{Q}$.

To Gibbs sample for the small $w_j$, consider the following procedure.
First prepare a maximally entangled state $\sum_{j\in[m]}\frac{\ket{j}\ket{j}}{\sqrt{m}}$ and apply $V'$ to the first register (and to some ancilla qubits) and set a flag to indicate success if the ancilla state \emph{does not remain} as $\ket{\bar{0}}$, so that the amplitude of $j$ decreases by a factor $\sqrt{1-\tilde{Q}^2(w_j)}$.
The flag signifies the part of the maximally entangled state that corresponds to the small $w_j$.
Then, using some fresh ancilla qubits $\ket{\bar{0}}$, apply the map $\bm{e}^{\frac{u_j-1}{2}}/4$ to the state $\ket{j}$ in the second register indicating ``success'' with a second flag qubit, apply this so that the additive error guarantee is $\frac{\xi}{6}$ if $u_j\leq 1$.
This can be done in $\bOt{\beta\log(1/\xi)}$ time and queries similarly to Lemma~\ref{lem:GibbsDense}.
The probability of measuring the success flag and seeing $\ket{j}$ in the second register is $\xi$-close to $(1-Q^2(w_j))\bm{e}^{u_j}/(16e m)$. Summing this over all $j$  amounts to an overall success probability of $\Omega(1)$ for this ``partial'' rejection sampling procedure, because $|u_j|\leq w_j\leq 1/2$ for at least half of the indices $j\in [m]$.

	Now to Gibbs sample for the large $w_j$'s we consider another procedure: prepare the state $\ket{\psi}$ and consider the first qubit being in the $\ket{0}$ state as ``success''.
    As before we apply $V'$ to the first register (and some ancilla qubits) but now we set a flag to indicate success if the ancilla state \emph{remains} $\ket{\bar{0}}$, so that the amplitude of $j$ decreases by a factor $\tilde{Q}(w_j)$. Now we would like to apply $\bm{e}^{(u_j-\tilde{u}_{\max})/2} / 8\sqrt{w_j} $, we do so in two steps.
	First we apply the map $\bm{e}^{(u_j-\tilde{u}_{\max})/2}/4$  to the state $\ket{j}$ in the last register with a additive error that is $\bigO{\xi}$ using $\bOt{\beta\log(1/\xi)}$ time and queries. We indicate success with a second flag qubit, similarly to Lemma~\ref{lem:GibbsDense}.
Then we apply a block-encoding $V''$ of $\diag(\tilde{R}(w))$ to the last register, where $|\tilde{R}(w_j)-1/\sqrt{4w_j}|=\bigO{\xi}$ for all $w_j\geq 1/2$, with the help of Lemma~\ref{lem:polyApx} and Theorem~\ref{thm:polyTransf}. This will take\textsuperscript{\ref{foot:QSVT}} $\bOt{\beta\log(1/\xi)}$ time and queries.
 We set a fourth success flag qubit if we retain the $\ket{\bar{0}}$ ancilla state.
The probability that we see all the flags indicating success and that we get measurement outcome $\ket{j}$ in the last register upon measuring the final state, is $\xi$-close to $ Q^2(w_j)\bm{e}^{u_j-\tilde{u}_{\max}}/(64 \beta s)$, which amounts to an overall success probability of $\Omega(\frac{1}{\beta s})$ for this ``partial'' rejection sampling procedure.
	
	Finally, notice that the two resulting partial Gibbs states are subnormalized in different ways. Let us define $N := 16em + 64 \beta s\bm{e}^{\tilde{u}_{\max}}$.
We sample a $j$ from the full distribution in the following way: with probability $\frac{16em}{N}$ sample $j$ using the first procedure, and with probability $\frac{64 \beta s\bm{e}^{\tilde{u}_{\max}}}{N}$ sample $j$ using the second procedure.
Since both procedures have total success probability $\Omega(\frac{1}{\beta s})$ the success probability of the final algorithm is also $\Omega(\frac{1}{\beta s})$.
Therefore we can rejection sample a $j$ with $\bigO{\sqrt{\beta s}}$ rounds of amplitude amplification with high probability, and can boost the success probability to $1-\delta/9$ by $\bigO{\log(1/\delta)}$ repetitions.
If we set $\xi=\bigO{\frac{\delta}{m\beta s}}=\bigO{\frac{\delta}{m^2}}$, then the distribution will be $\delta$ close to the true Gibbs distribution.
The complexity of each rejection sampling round is $\bOt{\beta\log(m/\xi)}$, which leads to the final complexity statement.
\end{proof}

\begin{theorem}\label{thm:SparseQuantum}
	Algorithm~\ref{alg:main} can be used to find an $\eps$-optimal pair of solutions with probability at least $1-\delta$ in $\frac{16}{\eps^2} \ln\left(\frac{nm}{\delta}\right)$ iterations. On a quantum computer with QCRAM the $t$-th iteration can be implemented in $\bOt{(1+t\eps)^\frac{3}{2}\sqrt{s}}$ time and the same number of quantum queries to a sparse oracle of $A$ with $s$-sparse rows and columns, giving  a total time $\&$ query  complexity $\bOt{\sqrt{s} / \eps^{3.5}}$.
\end{theorem}

\begin{proof}
	This follows from Lemma~\ref{lem:correctness} and Lemma~\ref{lem:GibbsDense}, by setting the additive error to $\bigO{\eps^2\delta/\ln(nm/\delta)}$ in the latter.
\end{proof}

\section{Reduction of general LP-solving to zero-sum games}\label{sec:reduction}
In this section we will reduce general LPs to zero-sum games to obtain a faster quantum algorithm for LP-solving. A similar argument can be found in for example~\cite{carmon2019Rank1Sketch}. We consider an LP in the following standard form:
\begin{align}
\max_{y\in \R^m} \ \ \ &b^Ty\nonumber\\
s.t.
\ \ \ & Ay \leq  c\label{eq:generalLPprimal}\\
& \phantom{A}y\geq 0, \nonumber
\end{align}
\vskip-3mm
\noindent which gives rise to the dual LP
\vskip-7mm
\begin{align}
\min_{x\in \R^n} \ \ \ &c^Tx\nonumber\\
s.t.
\ \ \ & A^Tx \geq  b\label{eq:generalLPdual}\\
& \phantom{A^T}x\geq 0. \nonumber
\end{align}
We assume without loss of generality that all entries of $A$, $c$ and $b$ are in $[-1,1]$. Furthermore, we assume a constant $R$ is known such that adding the constraint $\sum_{i=1}^m y_i \leq R$ to the primal will not change the optimal value of the primal. Similarly, we assume a constant $r$ is known such that adding the constraint $\sum_{j=1}^n x_j \leq r$ to the dual will not change the optimal value of the dual.  This implies that strong duality holds and the two optimal values coincide, so we will write $\opt$ for this value. Moreover, we can assume without loss of generality that $|c_i|\leq R$ for all $i\in [n]$, otherwise we could remove the constraint without changing the value of $\opt$. Finally, let $s$ denote a bound on the row and column sparsity of $A$, as well as on the sparsity of $b$ and $c$.
Our reduction will consist of the following steps:
\begin{enumerate}
\item Reduce an optimization problem to feasibility using binary search.
\item Scale the problem such that the solutions will be in the simplex.
\item Reduce to a problem where all right-hand sides of the inequality constrains are $0$.
\item Reduce to a zero-sum game.
\end{enumerate}

\paragraph{Reduction to feasibility.} Note that $-R \leq \opt \leq R$, because $\nrm{b}_\infty\leq1$ and there is an optimal solution with $\nrm{y}_1\leq R$.
Hence, if we can answer questions of the type ``is $\opt>\alpha$ or $\opt \leq \alpha+\eps$'' then we can use $\log(R/\eps)$ iterations of binary search to determine $\opt$ up to additive error $\eps$. To answer questions of this form we add $c^Ty \geq \alpha$ as a constraint and ask whether there is a feasible $y$. That is, we want to know whether there is no $y$ satisfying
\begin{align*}
  -b^Ty &\leq -\alpha\\
  Ay &\leq c\\
  \sum_i y_i &\leq R\\
  y &\geq 0
\end{align*}
or,~\footnote{These cases are potentially overlapping. In the intersection we are happy with either conclusion.} for some fixed $\eps'$, there is a $y$ such that
\begin{align*}
  -b^Ty &\leq -\alpha + \eps'\\
  Ay &\leq c+ e \eps'\\
  \sum_i y_i &\leq R\\
  y &\geq 0.
\end{align*}
Here we need to be careful with our choice of $\eps'$ since relaxing all the constraints by $\eps'$ might change the value of $b^Ty$ by as much as $\eps'(r+1)$, as we will show in the proof of Lemma~\ref{lem:reduction}. Hence we pick $\eps' = \Theta(\eps/ (r+1))$.
\paragraph{Scale to the simplex.} Note that by dividing all the right-hand sides by $R$ we simply scale down $y$ such that $\sum_i y_i \leq 1$. This, however, does imply that we want a lower additive error: $\eps'' =  \Theta(\eps/ (R(r+1)))$. Let us define $y'' := (y,z)^T$ where $z\in \R$ is a new variable. Then we obtain the new feasibility problem
\begin{align*}
  \exists ? y'' \in \Delta^{m+1}  \ \ s.t. & \\
                                A''y'' &\leq c''
\end{align*}
where
\[
A''= \left(\begin{array}{rr}
-b^T & 0\\
A & 0
\end{array}\right), \ \  c''= \left(\begin{array}{rr}
-\alpha / R \\
c / R
\end{array}\right)
\]

\paragraph{Right-hand sides zero.} For the final reduction to a zero-sum game, we want that all right-hand sides are equal to zero. To achieve this we use that $\sum_i y_i'' = 1$. We introduce an extra variable $h$ and add the constraint $\sum_i y_i'' = h$. If we now constrain the new vector $y''' = (y'',h)^T$ into the simplex instead of $y''$ then we find that $\sum_i y_i'' = h = 1/2$. Hence by scaling all the right-hand sides by a factor $1/2$ we will not have changed whether the LP is feasible. But now we have a variable $h$ that is fixed to a constant, so we can shift the inequalities by setting the right-hand side to zero while appropriately changing the coefficient of $h$ on the left-hand side. In particular we get the new feasibility problem
\begin{align*}
  \exists ? y''' \in \Delta^{m+2}  \ \ s.t. & \\
                                A'''y''' &\leq 0,
\end{align*}
where
\[
A'''= \left(\begin{array}{rr}
              e^T & -1\\
              \!-e^T & 1\\
A'' & -c''\kern-1.7mm
\end{array}\right),  \ \  y'''= \left(\begin{array}{ll}
y''\!\! \\
h
\end{array}\right).
\]
It suffices to bring down the additive error by a factor of at most two.
\paragraph{A zero-sum game.} To construct a zero-sum game as in~\eqref{eq:primal} we simply observe that
\begin{align*}
  \min_{y'''\in \Delta^{m+2}, \lambda \in \R} & \lambda \\
  \mathrm{s.t.} \ \ \  & A'''y''' \leq \lambda e
\end{align*}
has a value less than $\eps'''$ iff there is a point that is $\eps'''$-feasible for the last LP.

The final game matrix now is
\vskip-4mm
\begin{equation}\label{eq:generalgamematrix}
  A''' = \left(\begin{array}{rrr}
                 e^T & 1 & -1\\
                 -e^T & 1 & 1\\
                 -b^T & 0 & \alpha/R\\
                 A & 0 & -c/R
               \end{array}\right)
\end{equation}
Now let us prove that the sketch above indeed gives the desired result.

\begin{lemma}\label{lem:reduction}
  Finding the optimal value $\lambda^*$ of the game~\eqref{eq:generalgamematrix} up to additive error $\eps'''=\eps / (6R(r+1))$ suffices to correctly conclude either $\opt < \alpha$ or $\opt\geq \alpha-\eps$ for the corresponding LP~\eqref{eq:generalLPprimal}.
\end{lemma}
\begin{proof}
Finding the optimal value $\lambda^*$ of the game~\eqref{eq:generalgamematrix} up to additive error $\eps'''$ will tell us at least one of the following two things:
\begin{itemize}[itemsep=0mm]
\item $\lambda^* > 0$
\item $\lambda^* \leq 2 \eps'''$.
\end{itemize}

First assume we are in the case where $\lambda^* > 0$. In this case there is no $y''' \in \Delta^{m+2}$ such that
\[
A'''y'''\leq 0.
\]
On the other hand if we would have $\opt \geq \alpha$, then there would be a $y$ such that $Ay \leq c$, $\sum_i y_i\leq R, y \geq 0$ and $b^Ty\geq \alpha$. Let $\hat{y} = (y/(2R), 1/2-\sum_i y_i / (2R), 1/2)^T$. Then
\[
A'''\hat{y} = \left(\begin{array}{c} 0 \\ 0 \\ (-b^Ty+\alpha) / (2R) \\ (Ay - c)/ (2R) \end{array}\right) \leq 0
\]
which is a contradiction, hence $\opt < \alpha$.

Now we treat the other case: if $\lambda^* \leq 2\eps'''$ and $y^* = (y,z,h)$ is a strategy with this value, then we find that
\[
  A'''y^* \leq 2 \eps'''.
\]
This implies that $|h -1/2| \leq \eps'''$. Since $-b^T y + h\alpha/R \leq 2\eps'''$ we get that $b^T y \geq h\alpha/R - 2\eps''' \geq \alpha / (2R) - 2(\alpha/(2R) + 1)\eps'''$. A similar argument also shows that for all $j\in [n]$
\[
(A y)_j \leq  c_j/ (2R) + 2(c_j/(2R) + 1)\eps'''.
\] 
Let $\hat{y} = 2Ry$, then
\begin{alignat}{2}
  b^T\hat{y} &\geq  \alpha - 2(\alpha+2R) \eps''' &\geq \alpha - 6R\eps''',\label{eq:erroredopt}\\
  A\hat{y} &\leq  c + 2(c+2R) \eps''' &\leq c + 6R\eps'''.\label{eq:erroredconstraints}
\end{alignat}
Let $x^*$ be an optimal solution for the dual~\eqref{eq:generalLPdual}, such that $\sum_i x_i \leq r$. Then by applying weak duality on the relaxed constraints in~\eqref{eq:erroredconstraints} we find that
\[
  b^T\hat{y} \leq (c+e6R\eps''')^Tx^* \leq \opt + 6Rr\eps'''
\]
and hence by \eqref{eq:erroredopt} we can conclude
\[
  \opt \geq \alpha - 6R(r+1)\eps''' = \alpha - \eps. \qedhere
\]
\end{proof}

Via this reduction we give two new quantum LP-solvers. The first improves the error dependence of quantum LP-solvers to cubic; an LP-solver obtained from a state-of-the-art quantum SDP-solver~\cite{apeldoorn2018ImprovedQSDPSolving} would have a fifth-power error dependence. The second solver is based on our sparse algorithm and is the first quantum LP-solver that depends on the sparsity of the LP instead of on $n$ and $m$.
\begin{theorem}
  Given quantum query access to an LP of the form~\eqref{eq:generalLPprimal}, an $\eps$-optimal and $\eps$-feasible $y$ can be found with probability $1-\delta$ using either
  \begin{itemize}
  \item $\bOt{(\sqrt{n}+\sqrt{m})\left(\frac{R(r+1)}{\eps} \right)^{\!\!3} }$ quantum queries to a dense matrix oracle and the same number of other gates, or
  \item $\bOt{\sqrt{s}\left(\frac{R(r+1)}{\eps} \right)^{\!\!3.5} }$ quantum queries to a sparse matrix oracle and the same number of other gates.
  \end{itemize}
\end{theorem}
\begin{proof}
  The $\bOt{(\sqrt{n}+\sqrt{m})\left(\frac{R(r+1)}{\eps} \right)^3}$ bound follows directly from Lemma~\ref{lem:reduction} and Lemma~\ref{thm:DenseQuantum}.

For the sparse case let $s$ be the maximum of the sparsity of $A$, $b$ and $c$. Then apart from the first two rows, every row and column of $A'$ is $(s+3)$-sparse.
However, the row sparsity only matters in the complexity of the Gibbs-sampling step, in which multiples of the all-one vector can be added to the exponent without changing the Gibbs state. Since the first two rows are the all-one vector plus a $1$-sparse vector, we can treat them as effectively $1$-sparse for the Gibbs sampling step. The stated complexity bound then follows from Lemma~\ref{lem:reduction} and Lemma~\ref{thm:SparseQuantum}.
\end{proof}



\section*{Acknowledgments}
The authors are grateful to Ronald de Wolf for useful comments on the manuscript.

\bibliographystyle{alphaUrlePrint}
\bibliography{Bibliography}

\end{document}